\documentclass{article}

\usepackage[latin1]{inputenc}
\usepackage[T1]{fontenc}
\usepackage[english]{babel}
\usepackage[dvips]{graphicx}
\usepackage{enumerate}
\usepackage{amsthm}
\usepackage{amsmath}
\usepackage{amssymb}
\usepackage{color}
\usepackage{graphicx}
\usepackage{bbm}
\usepackage{float}
\usepackage{framed}
\usepackage{fullpage}
\input{new_macros.sty}

\newcommand{\F}{\mathbb{F}}
\renewcommand{\E}{\mathbb{E}}
\newcommand{\COMMENTZ}[1]{}

\newcommand{\Alice}{\textrm{Alice}}
\newcommand{\Bob}{\textrm{Bob}}

\newcommand{\FQ}{\F_Q}
\newcommand{\CHSH}{\mathbf{ \mathsf{CHSH}}}
\newcommand{\CHSHQ}{\mathbf{ \mathsf{CHSH_Q}}}
\newcommand{\OO}{\mathcal{O}}
\newcommand{\OOmega}{\mathrm{\Omega}}
\renewcommand{\mod}{\textrm{ mod }}
\newtheorem{CClaim}{Claim}

\newcommand{\Cq}{\omega(\CHSH_Q)}

\newcommand{\encad}[1]{%
	\fbox{\begin{minipage}{0.98\linewidth}%
			#1\end{minipage}}}

\usepackage{epstopdf,hyperref}

\newcommand{\be}{\begin{equation}}
\newcommand{\ee}{\end{equation}}

\usepackage{mathtools}
\def\multiset#1#2{\ensuremath{\left(\kern-.3em\left(\genfrac{}{}{0pt}{}{#1}{#2}\right)\kern-.3em\right)}}

\author{Rémi Bricout and André Chailloux \\ Inria, Paris.}
\title{Recursive cheating strategies for the relativistic $\F_Q$ bit commitment protocol}
\date{}
\begin{document}\thispagestyle{plain}\setcounter{page}{1}
	\maketitle
	\thispagestyle{plain}
\begin{abstract}
	In this paper, we study relativistic bit commitment, which uses timing and location constraints to achieve information theoretic security. We consider the $\F_Q$ multi-round bit commitment scheme introduced by Lunghi \etal \cite{LKB+15}. This protocol was shown secure against classical adversaries as long as the number of rounds $m$ is small compared to $\sqrt{Q}$ where $Q$ is the size of the used field in the protocol \cite{CCL15,FF16}.
	
	In this work, we study classical attacks on this scheme. We use classical strategies for the $\CHSH_Q$ game described in \cite{BS15} to derive cheating strategies for this protocol. In particular, our cheating strategy shows that if $Q$ is an even power of any prime, then the protocol is not secure when the number of rounds $m$ is of the order of $\sqrt{Q}$. For those values of $Q$, this means that the upper bound of \cite{CCL15,FF16} is essentially optimal. 
\end{abstract}

%\begin{keywords}
%	relativistic bit commitment, multi-round protocol, classical attacks.
%\end{keywords}

\section{Introduction}
\subsection{Context}
\thispagestyle{plain}
The goal of relativistic cryptography is to exploit the no superluminal signaling (NSS) principle in order to perform various cryptographic tasks. NSS states that no information carrier can travel at a speed greater than the speed of light. Note that NSS is closely related to the non-signaling principle that says that a local action performed in a laboratory cannot have an \emph{immediate} influence outside of the lab. NSS is more precise since it gives an upper bound on the speed at which such an influence can propagate. 
Apart from this physical principle, we want to ensure {information-theoretic} security meaning that the schemes proposed cannot be attacked by any classical (or quantum) computers, even with infinite computing power. This is in  contrast with used schemes, which most often rely on computational assumptions such as the hardness of factoring \cite{RSA}. 

The idea of using physical assumptions laws to ensure information theoretic security for cryptographic schemes is not a new one. The most striking example in recent years is Quantum Key Distribution (QKD) which allows two distant parties to distill a secret key with information-theoretic security \cite{BB84}. The main idea of QKD is to exchange quantum states on an insecure quantum channel and check a posteriori whether they have been disturbed. If not, it means that no eavesdropper was tampering with the quantum channel and the quantum states can be safely used to distill a secret. In fact, this works provided that the quantum states are not too noisy. QKD is quite practical and has indeed been widely deployed, but at the same time, it requires dedicated hardware and can only work today provided the 2 parties are not too far away from each other, at most a few hundred kilometers (see for instance \cite{KLH15} for the current record).

The idea of using the NSS principle for cryptographic protocols originated in a pioneering work by Kent in 1999 \cite{Kent99} as a way to physically enforce a non communication constraint between the different agents of one party (the idea of splitting up a party into several agents dates back to \cite{BGK88}, but without an explicit implementation proposal). The original goal of Kent was to bypass the no-go theorems for quantum bit-commitment \cite{May97,LC97}. Interestingly, this original protocol was classical and allowed for several rounds which increased the lifespan of the protocol. However, the protocol required to exchange messages whose length scaled exponentially in the number of rounds (i.e. the commitment time) and a feasible implementation was not possible for a large number of rounds.  A subsequent work \cite{Kent05} improved this scaling, but to our knowledge, no precise time/security tradeoff is available for this protocol.

More recently, quantum relativistic bit commitment protocols were developed where the parties exchange quantum systems, with the hope that combining the no superluminal signaling principle with quantum theory will lead to more secure (but less practical) protocols \cite{Kent11,Kent12, KTH13}. In particular, the protocol \cite{Kent12} was implemented in Ref.~\cite{LKB13}.  We note that the scope of relativistic cryptography is not limited to bit commitment. For instance, there was recently some interest (sparked again by Kent) for position-verification protocols \cite{KMS10, LL11,Unr14} but contrary to the case of bit commitment, it was shown that secure position-verification is impossible both in the classical and the quantum settings \cite{CGM09,BCF14}.

The original idea of \cite{BGK88} was recently revisited by Cr\'epeau \etal \cite{CSST11} (see also \cite{sim07}). Based on this work, Lunghi \textit{et al.}~devised a multi-round bit commitment protocol involving only four agents, two for Alice and two for Bob \cite{LKB+15}. They managed to prove that this protocol, which we call the ``$\FQ$ protocol'' from now on, remains secure for several rounds, against classical attacks. Unfortunately, this proof was rather inefficient since the complexity of the protocol (the size of the messages the agents need to exchange at each round) scaled exponentially with the number of rounds. 
Recently, two papers improved the security proof and showed that the complexity of the protocol in fact only scales logarithmically with the number of rounds \cite{CCL15,FF16}, implying that the commitment time is essentially unlimited:
\begin{theorem}[\cite{CCL15,FF16}]\label{Theorem:UpperBound}
	The $\FQ$ relativistic $m$-round bit commitment protocol is $\eps$-binding with $\eps = \OO(\frac{m}{\sqrt{Q}})$ against classical adversaries, meaning that Alice's cheating probability is at most $\frac{1}{2} + \OO(\frac{m}{\sqrt{Q}})$.
\end{theorem}

While the two proofs of this fact are very different, they rely to some extent on the analysis of $\CHSHQ$, a non-signaling game that generalizes the well-known $\CHSH$ game to the case where inputs and outputs are not restricted to being bits, but rather belong to $\FQ$ the Galois Field of order $Q$.

Notice that in the way the cheating probability is defined, a perfectly secure protocol will have cheating probability of $\frac{1}{2}$ for both Alice and Bob.  So an $\eps$-secure (here $\eps$-binding) protocol will have a cheating probability of $\frac{1}{2} + \eps$. The protocol has (stand-alone) security when $\eps$ is small.

The above result shows that the protocol is secure as long as $m \ll \sqrt{Q}$ but it was not known for larger values of $Q$, in particular when $m$ approaches, or even exceeds $\OOmega(\sqrt{Q})$. Very recently, this protocol has been implemented by keeping the agents $7$ km apart and demonstrated a sustain period of $24$ hours \cite{VMH+16}. Also, it is important to know that the number of bits sent at each round is $\log(Q)$ and therefore $Q$ can be efficiently made exponentially big in the security parameter.

 Until now, no cheating strategy has been proposed for this scheme.

\COMMENTZ{A recent result, published in the quant-ph arXiv, claims to improve the above result as follows 
\begin{CClaim}[\cite{PPP16}]
	The $\FQ$ relativistic $m$-round bit commitment protocol is $\eps$-binding with $\eps = \OO(\sqrt{\frac{m}{{Q}}})$ against classical adversaries, meaning that Alice's cheating probability is at most $\frac{1}{2} + \OO(\sqrt{\frac{m}{{Q}}})$.	
\end{CClaim}}

\subsection{Contributions}
Our main contribution is to present the first attack on the $\F_Q$ protocol. We show the following 
\begin{theorem}\label{Theorem:Main}
	There exists an attack on the $m$-round $\F_Q$ protocol in which Alice's cheating probability is 
	$$ 1 - \dfrac{1}{2}\left(\left( 1-\dfrac{1}{Q} \right)\left( 1-\Cq \right)\right)^{\lfloor\frac{m-1}{3}\rfloor} $$
	where $\omega(\CHSH_Q)$ is the classical value of the $\CHSH_Q$ game and $m \ge 3$. 
\end{theorem}

In \cite{BS15}, it was shown for any prime $p$ and integer $n$ that 
$$ \omega(\CHSH_Q) = \left\{
\begin{array}{ll}
\OOmega(\sqrt{\frac{1}{Q}}) \quad & \textrm{if } Q = p^{2n} \\
\OO(Q^{- \frac{1}{2} - \eps_0}) \quad & \textrm{if } Q = p^{2n + 1} \\
\end{array}
\right. $$
where $\eps_0$ is a constant proven to be positive.

In particular, Theorem \ref{Theorem:Main}, combined with the above bound, shows that the upper bound of \cite{CCL15,FF16} (Theorem \ref{Theorem:UpperBound}) is essentially optimal when considering an even prime power :
\begin{corollary}
	For any integer $n$, prime number $p$ and $Q = p^{2n}$ there is a cheating strategy for Alice that achieves success probability 
	$$ 1 -  \frac{1}{2}\left(1 - \OOmega(\frac{1}{\sqrt{Q}})\right)^{\lfloor\frac{m - 1}{3}\rfloor}. $$
	\begin{itemize}
		\item If $m \ll \sqrt{Q}$ then the above cheating probability is equal to $\frac{1}{2} + \OOmega(\frac{m}{\sqrt{Q}}) - o(\frac{m}{\sqrt{Q}})$.
		\item If $m = t \sqrt{Q}$ then the above cheating probability is lower bounded by $1 - \frac{1}{2} e^{-\frac{t}{3}}$, which quickly approaches $1$ as $t$ increases.
	\end{itemize}
\end{corollary}
From the above bounds, we can conclude that up to constant factors, our attack is optimal when $Q$ is an even power of a prime. 

We note also that there is an upper bound on the value of $\CHSH_Q$ when $Q$ is an odd power of a prime. In this case, we have $\omega(\CHSH_Q) = \OOmega(Q^{-2/3})$ \cite{BS15,PP16}. From there, we have

\begin{corollary}
	For any integer $n$, prime number $p$ and $Q = p^{2n + 1}$ there is a cheating strategy for Alice that achieves success probability 
	$$ 1 - \frac{1}{2}\left(1 - \OOmega(Q^{-2/3})\right)^{\lfloor\frac{m-1}{3}\rfloor}. $$
\end{corollary}	
	
	 and if $m = tQ^{\frac{2}{3}}$ then Alice can cheat with probability $1 - \frac{1}{2}e^{-\frac{t}{3}}$ which quickly converges to $1$ as $t$ increases. \\

This result also shows that even an improved bound on $\omega(\CHSHQ)$ variants presented in \cite{PPP16} cannot be used to improve - except in the constants - the security of the $\FQ$ protocol, at least for even prime powers of $Q$. \\

Our second contribution is an extension of this attack to more realistic scenarios from the attacker's point of view. 
In the relativistic model, we assume that cheating Alice can perform communications between $\mathcal{A}_1$ and $\mathcal{A}_2$ such that both agents of Alice know exactly the whole transcript of the protocol, except the last round message sent to the other Alice. Proving security in this setting allows us to minimize the spacetime requirements in order to achieve security.

However, our attack also assumes this power for cheating Alice and this could be very challenging in practice. Therefore, we introduce the notion of propagation time which corresponds to the number of rounds $\rho$ that can pass until the Alices are able to send some information to one another. In the original model, this propagation time is $2$ rounds. We perform extend Theorem \ref{Theorem:Main} to the following setting
\begin{itemize}
	\item The propagation time $\rho$ can be larger than $2$.
	\item The two Alices know the bit they want to reveal only after $k_0 \geq 1$ rounds. We call $k_0$ the decision time.
\end{itemize}
Showing that the attack works in this setting ensures that simple countermeasures consisting of increasing the distance between the two pairs will not significantly reduce the efficiency of the attack. We show the following:

\begin{theorem}\label{Theorem:PropagationTime}
	For any propagation time $\rho \ge 2$, and any decision time $k_0$, there exists an attack on the $m$-round $\F_Q$ protocol where Alice's cheating probability is 
	$$ 1 - \dfrac{1}{2}\left(( 1-\dfrac{1}{Q})\left( 1-\omega(\CHSHQ) \right)\right)^{\lfloor\frac{m - k_0 - 1}{\rho + 1}\rfloor}. $$
	for $m \ge k_0 + 2$.
\end{theorem}

\noindent {\bf Organisation} --- In Section \ref{Section:Preliminaries}, we present the $\F_Q$ protocol as well as the $\CHSH_Q$ game. In Section \ref{Section:TheAttack}, we present our main result, namely the attack on the $\F_Q$ protocol. Finally, in Section \ref{Section:AttackExtension}, we present the extension of this attack to more realistic scenarios.

\section{Preliminaries}\label{Section:Preliminaries}
\subsection{Bit commitment}
\emph{Bit commitment} is a cryptographic primitive between two distrustful parties Alice and Bob which consists of $2$ phases: a \emph{Commit phase} and a \emph{Reveal phase}. Alice has a random bit $d$ at the beginning of the protocol. In the commit phase, Alice will commit to this value $d$ by performing some communication protocol such that at end of the commit phase, Bob knows no information about $d$. In the second phase, the reveal phase, Alice and Bob also perform some communication which results in Alice revealing $d$.  A desired property here is that Alice is unable to reveal a bit different from the one chosen during the commit phase. 

In some sense, a bit commitment protocol simulates a digital safe. In the commit phase, Alice writes her input $d$ on a piece of paper, puts that paper into the safe and sends the safe to Bob. If Bob has no information about the key safe then he cannot open it and therefore has no information about $d$. In the reveal phase, Alice would send to Bob the key to open the safe. But she cannot change the value of the bit in the safe because Bob has control of the safe. This primitive has been widely studied. However,  bit commitment can only be performed with computational security in most usual models. 

%We now define more formally a bit commitment scheme.
We now give a formal definition of the bit commitment scheme.

\begin{definition}
	A bit commitment scheme  is an interactive protocol between Alice and Bob with two phases, a Commit phase and a Reveal phase.
	
	\begin{itemize}
		\item \emph{Commit phase}. Alice chooses a uniformly random input $d$ that she wants to commit to. To do so, Alice and Bob perform a communication protocol that corresponds to this commit phase. 
		\item \emph{Reveal phase}. Alice interacts with Bob in order to reveal $d$. To do so, they perform a second communication protocol where at the end, Bob should know the value revealed by Alice. Bob, depending on this revealed value and the interaction with Alice, outputs ``Accept" or ``Reject".  
	\end{itemize}
\end{definition}
\noindent We also define the following security requirements for the commitment scheme.
\begin{definition}
	A bit commitment protocol is said to be correct if when both players are honest, Bob never outputs ``Reject".
\end{definition}
A cheating strategy $S$ for Alice can be therefore decomposed into a cheating strategy $S_{commit}$ for the commit phase and $S_{reveal}$ for the reveal phase and we will usually write $S = (S_{commit},S_{reveal})$. The goal of a cheating Alice is to choose the value she wants to reveal only after the commit phase. The reveal strategy $S_{reveal}$ will depend on the value $d$ she wants to reveal. We denote by $S_{reveal}(d)$ Alice's cheating strategy in the reveal phase for a fixed $d$.

\begin{definition}
	For a fixed cheating strategy $S = (S_{Commit},S_{reveal})$ for Alice, we define Alice's cheating probability $P^*_A(S)$ as
	\begin{align*} 
	P^*_A(S) & := 
	\frac{1}{2} \Pr[\mbox{ Alice successfully reveals } d = 0 | (S_{Commit},S_{reveal}(0))] + \\
	& 	\qquad \frac{1}{2} \Pr[\mbox{ Alice successfully reveals } d = 1 | (S_{Commit},S_{reveal}(1)) ].
	\end{align*}
\end{definition}

\begin{definition}
	We define Alice's optimal cheating probability $P^*_A$ as 
	$$ P^*_A := \max_{S = (S_{Commit},S_{reveal})} P^*_A(S).  $$
	We say that a bit commitment is $\eps$ sum-binding if $P^*_A \le \frac{1}{2} + \eps$.
\end{definition}

Here, we used one of several possible definitions for the binding property. This definition is weak, since it doesn't necessarily behave well under composition. In order to prove security, even for relativistic bit commitment protocols, some stronger definitions of security are used (see for example \cite{FF16}). While using a stronger security definitions strengthens upper bounds on the cheating probability, it is by using the weakest security definition that we have the strongest lower bounds on those cheating probabilities. Since in this paper, we present cheating strategies, $\ie$ lower bounds, we use the weak notion of sum-binding. 

Another security condition we want to ensure is the hiding property. At the end of the commit phase, we don't want Bob to have a lot of information about the committed bit $d$. This means that to ensure the hiding property, we will only be interested in a cheating Bob's strategy during the commit phase, and a cheating strategy $S^B$ for Bob will be a strategy that he will use to try to learn $d$ after the commit phase. 

\begin{definition}
	For a fixed cheating strategy $S^B$ for Bob, we define his cheating probability $P^*_B(S^B)$ as 
	$$ P^*_B(S^B) := \Pr[\mbox{Bob guesses } d \mbox{ after the Commit phase} | S^B]. $$
\end{definition}

\begin{definition}
	We define Bob's optimal cheating probability $P^*_B$ as 
	$$ P^*_B := \max_{S^B} P^*_B(S^B).  $$
	We say that a bit commitment is $\eps$-hiding is $P^*_B \le \frac{1}{2} + \eps$.
\end{definition}
\subsection{Relativistic bit commitment}\label{Section:DescriptionOfTheProtocol}
A relativistic bit commitment scheme is a commitment scheme where we use physical property that no information carrier can travel faster than the speed of light. In order to take advantage of this principle, we split Alice (resp. Bob) into $2$ agents $\mathcal{A}_1$ and $\mathcal{A}_2$ (respectively $\mathcal{B}_1$ and $\mathcal{B}_2$). For each $i \in \{1,2\}$, $\Alice_i$ interacts only with $\Bob_i$. If we put the two pairs $(\mathcal{A}_1,\mathcal{B}_1)$ and $(\mathcal{A}_2,\mathcal{B}_2)$ far apart, and use some timing constraints, we can create some non-signaling type scenarios. Here, we will only use the property that the two honest Bobs know their respective location. In particular, there is no trust needed regarding the location of the cheating parties. 

The security definitions for relativistic bit commitment are the ones we presented in the above Section. We will now describe the $\F_Q$ relativistic bit commitment scheme. This description will consist of $4$ phases, the preparation phase, the commit phase, the sustain phase and the reveal phase. The preparation phase is some preprocessing phase that can be done anytime before the protocol. The sustain phase can be seen as a part of the reveal phase, and corresponds to the time where the committed bit is safe. We assume here that the two Alices know at the beginning of the sustain phase the bit $d$ they want to reveal. In Section \ref{Section:AttackExtension}, we will relax this requirement.

{\bf The single-round $\F_Q$ protocol}. --- The single-round version corresponds to the protocol introduced by Cr{\'e}peau \etal \cite{CSST11} (see also \cite{sim07}).
Both players, Alice and Bob, have agents $\mathcal{A}_1, \mathcal{A}_2$ and $\mathcal{B}_1, \mathcal{B}_2$ present at two spatial locations, 1 and 2, separated by a distance $D$. We consider the case where Alice makes the commitment. 
The protocol (followed by honest players) consists of 4 phases : preparation, commit, sustain and reveal. The sustain phase in the single-round protocol is trivial and simply consists in waiting for a time less than $D/c$, which is the time needed for light to travel between the two locations. The bit commitment protocol goes as follows. 
\begin{enumerate}
	\item \emph{Preparation phase}: $\mathcal{A}_1,\mathcal{A}_2$ (resp.~$\mathcal{B}_1,\mathcal{B}_2$) share a random number $a\in \FQ$ (resp.~$x \in \FQ$).
	\item \emph{Commit phase}: $\mathcal{B}_1$ sends $x$ to $\mathcal{A}_1$, who returns $y = a + d \cdot x$ where $d \in \zo$ is the committed bit. 
	\item \emph{Sustain phase}: $\mathcal{A}_1$ and $\mathcal{A}_2$ wait for some time $\tau < D/c$.
	\item \emph{Reveal phase}: $\mathcal{A}_2$ reveals the values of $d$ and $a$ to $\mathcal{B}_2$ who checks that $y = a + d \cdot x$.
\end{enumerate}

{\bf The multi-round protocol}.--- \label{prot}
The protocol above was recently extended to a multi-round commitment scheme \cite{LKB+15}. The main idea to increase the commitment time is to delay the reveal phase and have $\mathcal{A}_2$ commit to the string $a$ instead of revealing it. In fact, the new sustain phase will now consist of many rounds where the active players (i.e. the player to commits in that given round and the corresponding player for Bob) alternate between locations 1 and 2, separated by a distance $D$. The $m$-round bit commitment protocol goes as follows 
\begin{enumerate}
	\item \emph{Preparation phase}: $\mathcal{A}_1,\mathcal{A}_2$ (resp.~$\mathcal{B}_1,\mathcal{B}_2$) share $m$ random numbers $a_1,\dots,a_{m}$ (resp.~$x_1,\dots,x_{m}$) $\in \FQ$. 
	\item \emph{Commit phase}: $\mathcal{B}_{1}$ sends $x_{1}$ to $\mathcal{A}_{1}$, who returns $y_{1} = d \cdot x_{1} + a_{1}$, where $d \in \{0, 1\}$ is the committed bit.
	\item \emph{Sustain phase}: for each round $k$, with $2 \le k < m$,  $\mathcal{B}_{k\mod 2}$ sends $x_{k}$ to $\mathcal{A}_{k\mod 2}$, who returns $y_{k} = x_{k}\cdot a_{k-1} + a_{k}$.
	\item \emph{Reveal phase}: $\mathcal{A}_1$ reveals $d$ and $y_m = a_{m-1}$ to $\mathcal{B}_1$. Bob checks that $y_m = \alpha_{m-1}$, where we recursively define $\alpha_0 := d$, $\alpha_i := y_i - b_i * \alpha_{i-1}$. $\alpha_i$ corresponds to what $a_i$ ``should be''. 
\end{enumerate}

The main idea of the multi-round protocol is to delay the reveal phase in order to increase the commitment time. This delay is obtained by making the passive Alice commit to the value of the string she was supposed to reveal in the previous round. Since each round increases the total commitment time by $D/c$, modulo the time needed for the various algebraic manipulations in $\FQ$, one sees that the required number of rounds scales linearly with the commitment time one wishes to achieve.

We require that round $j$ finishes before any information about $x_{j-1}$ reaches the other Alice. For any $j$, we therefore have the following : active Alice has no information about $x_{j-1}$. This means that $y_j$ is independent of $x_{j-1}$. This will be crucial in order to show security of the protocol. One important thing to notice is that $d$, the bit Alice wants to reveal can be decided just after the commit phase. Therefore, $y_1$ is independent of $d$ but all the other messages $y_2,\dots,y_m$ can depend on $d$. 

Both those protocol are perfectly hiding. Moreover, from Theorem \ref{Theorem:UpperBound}, the multi-round protocol is $\eps$ sum-binding, with $\eps = O(\frac{m}{\sqrt{Q}})$.

\subsection{The $\CHSH_Q$ game}
A crucial tool for our attack (and for the previous security analysis), is the $\CHSH_{Q}$ two-player game introduced by Buhrman and Massar \cite{BM05}. This game is a natural generalization of the CHSH game to the field $\F_Q$, where two cooperating but non-communicating parties, Alice and Bob, are respectively given an input $x$ and $y$ chosen uniformly at random from $\F_Q$, and must output two numbers $a, b \in \F_Q$. They win the game whenever the condition $a + b = x \cdot y$ is satisfied. The value of a game $G$, denoted $\omega(G)$,  corresponds to the maximum probability of winning the game. A recent result by Bravarian and Shor \cite{BS15} establishes  bounds on $\omega(\mathrm{CHSH}_Q)$. They show the following

\begin{proposition}
for any prime $p$ and integer $n$, we have 
$$ \omega(\CHSH_Q) = \left\{
\begin{array}{ll}
\OOmega(\sqrt{\frac{1}{Q}}) \quad & \textrm{if } Q = p^{2n} \\
\OO(Q^{- \frac{1}{2} - \eps_0}) \quad & \textrm{if } Q = p^{2n + 1} \\
\end{array}
\right. $$
for some absolute constant $\eps_0>0$.
\end{proposition}

We define a variant of the $\CHSH_Q$ game, that we call $\CHSH_Q^\gamma$, which will be well defined for any $\gamma \in [0,1]$. We will use this variant in Section \ref{Section:AttackExtension}, when we will have longer propagation and decision times.
\begin{definition}
	In $CHSH_Q^\gamma$, Alice receives $x = 0$ with probability $\gamma$ and a random element $x \in \F_Q^* $, each with probability $\frac{1 - \gamma}{Q-1}$. Bob receives an input $y$ according to the same probability distribution. They output respectively $a$ and $b$ in $\F_Q$ and they win the game iff. $a + b = x\cdot y$.
\end{definition}
In short, $\CHSH_Q^\gamma$ is the same game as $\CHSH_Q$, but the input distribution of each player is slightly biased towards $0$. We have by definition $\CHSH_Q^{\frac{1}{Q}} = \CHSH_Q$. When playing $\CHSH_Q^\gamma$, we have:
	\begin{itemize}
		\item The probability that Alice and Bob get $(0,0)$ is $\gamma^2$.
		\item The probability that they get an element $(0,i)$ or $(i,0)$ with $i \in \F_Q^*$ is equal to $\frac{\gamma(1-\gamma)}{Q-1}$ for each such element.
		\item The probability that they get an element $(i,j)$ with $i,j \in \F_Q^*$ is equal to $\frac{(1-\gamma)^2}{(Q-1)^2}$ for each such element.
	\end{itemize}

	Inspired by shift techniques used in \cite{BS15}, we can show:
\begin{lemma}\label{Lemma:CHSH_Q^delta}
For any $\gamma \in [0,1]$, $\omega(CHSH_Q^\gamma) \ge \omega(\CHSH_Q)$.
\end{lemma}
\begin{proof}
As randomized strategies are nothing more than linear combinations of deterministic strategies, of which winning probability is given by the same linear combination, we can assume that all used optimal strategies are deterministic without loss of generality.

We consider an optimal strategy $S = (s_{1}, s_{2})$ for the $\CHSH$ game i.e. function $s_1,s_2 : \F_Q \rightarrow \F_Q$ such that $\Pr_{x, y}[s_{1}(x)+s_{2}(y)=xy] = \omega(\CHSHQ)$, where the probability is over uniform $x$ and $y$. We define $p_{x,y} := 1$ if $s_{1}(x)+s_{2}(y)=xy$ and $0$ otherwise, which implies $\E_{xy} p_{xy} = \omega(\CHSH_Q)$. Let
	$$ Z_{u,v} := \gamma^2 p_{u,v} + \frac{\gamma(1-\gamma)}{Q - 1} \left(
	\sum_{x \in \F_Q - \{u\}} p_{xv} + \sum_{y \in \F_Q - \{v\}} p_{uy}\right) + \frac{(1-\gamma)^2}{(Q-1)^2} \sum_{\substack{x \in \F_Q - \{u\} \\ y \in \F_Q - \{v\}}} p_{xy}.$$
	
	$Z_{u,v}$ corresponds to the probability of winning the game $\CHSHQ$ on a changed probability distribution. In particular, $Z_{0,0}$ corresponds to the probability of winning $\CHSH_Q^\gamma$ when using strategy $S$.
	One can check that $\E_{u,v} [Z_{u,v}] = \omega({\CHSHQ})$, so we can fix a pair $(u,v)$ such that $Z_{u,v} \ge \omega(\CHSHQ)$.
	
	We now consider the strategy $S' = (s'_{1}, s'_{2})$ where $s'_{1}(x) = s_{1}(x+u) - xw$ and $s'_{2}(y) = s_{2}(y+v) - yu - uv$. $S'$ wins for $(x, y)$ precisely when $S$ wins for $(x+u, y+v)$. Indeed :
	$$\begin{array}{lll}
	s'_{1}(x)+s'_{2}(y) = xy & \Leftrightarrow & s_{1}(x+u) - xv + s_{2}(y+u) - yu - uv = xy\\
	& \Leftrightarrow & s_{1}(x+u) + s_{2}(y+v) = (x+u)(y+v)
	\end{array} $$
	
	Similarly, as before, we define $p'_{xy} = 1$ if $s'_1(x) + s'_2(y) = x\cdot y$ and $0$ otherwise. From the above equivalence, we have $p'_{x,y} = p_{(x+u),(y+v)}$. We also define 
	$$ Z'_{u,v} := \gamma^2 p'_{u,v} + \frac{\gamma(1-\gamma)}{Q - 1} \left(
	\sum_{x \in \F_Q - \{u\}} p'_{xv} + \sum_{y \in \F_Q - \{v\}} p'_{uy}\right) + \frac{(1-\gamma)^2}{(Q-1)^2} \sum_{\substack{x \in \F_Q - \{u\} \\ y \in \F_Q - \{v\}}} p'_{xy}.$$
	Notice that $Z'_{0,0}$ corresponds to the probability of winning $\CHSH_Q^\gamma$ when using strategy $S'$. Moreover, for any $(x,y)$, we have $Z'_{x,y} = Z_{x+u,y+v}$. From there, we conclude 
	$$ \omega(\CHSH_Q^\gamma) \ge Z'_{0,0} = Z_{u,v} \ge \omega(\CHSH_Q), $$
	which proves the desired result.
\end{proof}
\section{Attack with perfect conditions}\label{Section:TheAttack}

In this Section, we present our construction of a cheating strategy which will be essentially optimal for some values of $Q$. The protocol is perfectly hiding. Therefore, we are only interested in the binding property, \ie in cheating Alice.

The idea of the attack is the following. Every three rounds (or more in Section \ref{Section:AttackExtension}), Alice's agents have an occasion to play a $\CHSH_Q$ game. If they win this game, which happens with probability $\omega(\CHSH_Q)$, they can easily fool Bob (with the provided strategy, sending only zeros until \emph{reveal phase} is fine). If they do not manage to win the $\CHSH_Q$ game, they just try again three rounds later. More precisely, for each step of three rounds, the last two rounds are used to play the $\CHSH_Q$ game. The first one allows $\mathcal{A}_1$ and $\mathcal{A}_2$ to determine if they won during the previous step, or calculate a corrective factor $\eta$ if they did not. Thus, for a m-round long protocol, Alice's agent can play roughly $\frac{m}{3}$ such $\CHSH_Q$ games. As it is sufficient for them to win one of these games, we see that cheating probability grows exponentially with the number of rounds. Moreover, at each of these sets of three rounds, an additional factor $(1-\frac{1}{Q})$ appears. Indeed, if at the third round of the set Bob sends a $0$, Alice is also in a situation in which she can cheat, because this $0$ makes Alice's error collapse to zero. However, the contribution of this additional factor can be neglected (it is only $\OO(\frac{1}{Q})$, compared to the $\OO(\frac{1}{\sqrt{Q}})$ given by $\CHSH_Q$).

We assume for now that the propagation time of the information is $2$ rounds. This means that when $\Alice_{(i \mod 2)}$ receives $x_i$, the other Alice will know the value of $x_i$ at round $i+2$. Therefore, a cheating strategy for Alice is described by a $m$-tuple of functions $S = (s_1,\dots,s_m)$, where each $s_i$ corresponds to Alice's output function at round $i$. $s_1$ is a function of $x_1$ and $s_i$ is a function of $(x_0,\dots,x_{i-2},x_i)$ for $i \ge 1$ where we use the convention $x_0 = d$. For each $i \ge 1, x_i \in \F_Q$ and the output space of each $s_i$ is $\F_Q$.

Consider any fixed cheating strategy $S$ for Alice. At the end of the protocol, Bob checks that $y_m = \alpha_{m-1}$. When we expand $\alpha_{m-1}$ as a function of  $(d,x_1,\dots,x_{m-1})$, the checking condition, that we call $\mathcal{C}_m$ becomes  

$$y_{m} = y_{m-1} - x_{m-1}\Bigg(y_{m-2}-x_{m-2}\bigg(...\;...-x_{2}(y_{1}-d\cdot x_{1})...\bigg)\Bigg).$$

If this equality is not verified, Alice is caught cheating. On the other hand, if $\mathcal{C}_m$ is verified then Bob cannot distinguish an honest Alice from a dishonest one, and he does not abort. 

Let $\mathcal{C}_m(S, d,x_1,\dots,{x_{m-1}})$ the event which corresponds to the above equality being verified. Alice's cheating probability using $S$, that we note $g_m(S)$ is therefore
$$ g_m(S) := \Pr_{d,x_1,\dots,x_{m-1}}[\mathcal{C}_m(S, d,x_1,\dots,x_{m-1})]$$  
where $d$ is a uniformly random bit and $x_1,\dots,x_{m-1}$ are uniformly random elements of $\F_Q$. We also define $g_m := \max_{S} g_m(S)$ which is Alice's maximal cheating probability in $\mathcal{P}_m$. In this section, we present a cheating strategy $S$ for Alice such that 
$$g_m(S) = \frac{1}{2} + \frac{1}{2}\left(1 - \left(1 - \omega(\CHSH_Q)\right)^{\lfloor\frac{m - 1}{3}\rfloor}\right). $$
which will prove Theorem \ref{Theorem:Main}. In order to do so, we first modify protocol $\mathcal{P}_m$ to make it more symmetric (Section \ref{Section:Symmetrization}). Then, we describe our attack (Section \ref{Section:DescriptionOfTheAttack}) and we prove its cheating probability (Section \ref{Section:MainProof}).

\subsection{Symmetrization of the protocol}\label{Section:Symmetrization}

We want to describe a recursive strategy for protocol $\mathcal{P}_{m}$.
Unfortunately, this protocol induces a difference between Alice's strategy at round $1 \leq k < m$ and her strategy at round $m$. Because of that, it is difficult to study the protocol recursively. 

We therefore consider a modified protocol $\mathcal{P'}_{m}$, which, as we will show, is a bit easier than $\mathcal{P}_{m}$ to win, but harder than $\mathcal{P}_{m+1}$. In this modified version, at round $m$, $\Bob_{m \mod 2}$ sends an additional random string $x_{m} \in \F_Q$, and $\Alice_{m \mod 2}$ returns $y_{m} = x_{m}\cdot a_{m-1}$ instead of $y_{m} = a_{m-1}$. All other rounds are unchanged. Similarly, as for $\mathcal{P}_m$, a cheating strategy for Alice $S'$ can be described as a $m$-tuple of functions $(s'_1,\dots,s'_m)$ that give Alice's outputs $y_i$ depending on her accessible information at round $i$. 

Bob checks now that $y_m = x_m \cdot \alpha_{m-1}$ and therefore, the condition Alice must satisfied to win is modified into $\mathcal{C'}_m(S', d,x_1,\dots,x_m)$, where 
$$\mathcal{C'}_m(S', d,x_1,\dots,x_m) \Leftrightarrow y_{m} = x_{m}\Bigg(y_{m-1} - x_{m-1}\bigg(y_{m-2}-x_{m-2}\Big(...\;...-x_{2}(y_{1}-d\cdot x_{1})...\Big)\bigg)\Bigg)$$
By expanding $\mathcal{C'}_m(S', d,x_1,\dots,x_m)$, it can be written down as :
$$\begin{array}{llll}
y_{m}&=&&x_{m}\cdot y_{m-1}\\
&&-&x_{m}\cdot x_{m-1}\cdot y_{m-2}\\
&&+&x_{m}\cdot x_{m-1}\cdot x_{m-2}\cdot y_{m-3}\\
&&\vdots&\vdots\\
&&-&(-1)^{m}x_{m}\cdot x_{m-1}\cdot x_{m-2}\cdot ...\cdot x_{1} \cdot d
\end{array} $$
or, using a compact form:
$$\mathcal{C'}_m(S', d,x_1,\dots,x_m) \Leftrightarrow y_{m} = \sum\limits_{i=1}^{m-1} \left( (-1)^{m-i}y_{i}\cdot\prod\limits_{j=i+1}^{m}x_{i} \right) - (-1)^{m}d\cdot\prod\limits_{j=1}^{m}x_{j}$$
For a cheating strategy $S'$, Alice's winning probability $g'_m(S')$ for this modified protocol is therefore defined as  
$$ g'_{m}(S') := \Pr_{d,x_1,\dots,x_m}[\mathcal{C'}_m(S', d,x_1,\dots,x_m)] \quad \textrm{and} \quad g'_m := \max_{S'} g'_m(S')$$

We show the following

\begin{lemma}\label{Lemma:Symmetrization}
	$\forall m \geq 2$, we have $g_{m} \leq g'_{m} \leq g_{m+1}.$
\end{lemma}

\begin{proof}~
	\begin{itemize}
		\item For the first inequality, let us consider the optimal strategy $S = (s_{1}, ..., s_{m})$ for $\mathcal{P}_{m}$, where $s_{k}$ is Alice's strategy at round $k$ (i.e. a function that outputs $y_{k}$ when given Alice's knowledge at round $k$). Alice's cheating probability for $\mathcal{P}_m$ is $g_m(S)$. Consider the following strategy $S' := (s_{1}, ..., s_{m-1}, s'_{m})$ for $\mathcal{P}'_m$, where $s'_{m}(d,x_1,\dots,x_{m-2},x_m) := x_{m} \cdot s_{m}(d, x_1,\dots,x_{m-2})$. $S'$ allows to win on $\mathcal{P'}_{m}$ at least as efficiently as $S$ on $\mathcal{P}_{m}$, because $S'$ wins whenever $S$ does. Indeed, suppose that $S$ is a winning strategy for a given $(d, x_{1}, ..., x_{m-1})$. This means that $\mathcal{C}_m(S, d,x_1,\dots,x_{m-1})$ is satisfied or equivalently:
		$$s_{m}(d,x_1,\dots,x_{m-2}) = y_{m-1} - x_{m-1}\Big(y_{m-2}-...\;...-x_{2}(y_{1}-d\cdot x_{1})...\Big)$$
		Then, since $s'_{m}(d, x_1, \dots, x_{m-2},x_m)  = x_{m} \cdot s_{m}(d,x_1,\dots,x_{m-2})$, we get
		$$s'_{m}(d, x_1, \dots, x_{m-2},x_m) = x_{m}\bigg(y_{m-1} - x_{m-1}\bigg(y_{m-2}-...\;...-x_{2}(y_{1}-d\cdot x_{1})...\Big)\bigg)$$
		which implies $\mathcal{C}'_m(S', d,x_1,\dots,x_m)$, for any $x_m$. From there, we immediately get 
		$$g_m = g_m(S) \le g'_m(S') \le g'_m.$$
\item For the other inequality, we fix an optimal strategy $S' = (s_{1}, ..., s_{m})$ for $\mathcal{P'}_{m}$. We consider the following strategy $S := (s_{1}, ..., s_{m}, \overline{0})$ for $\mathcal{P}_{m+1}$, where $\overline{0}$ is the function that always outputs $0$, no matter the inputs. This means that when performing $S$, we always have $y_{m+1} = 0$. $S$ is at least as good to win $\mathcal{P}_{m+1}$ as $S'$ is to win $\mathcal{P'}_{m}$. Indeed, if for a tuple $(d, x_{1}, ..., x_{m})$, $S'$ wins on $\mathcal{P'}_{m}$, then $\mathcal{C'}(S', d,x_1,\dots,x_m)$ holds or equivalently
		$$y_{m} = x_{m}\Bigg(y_{m-1} - x_{m-1}\bigg(y_{m-2}-x_{m-2}\Big(...\;...-x_{2}(y_{1}-d\cdot x_{1})...\Big)\bigg)\Bigg)$$
		From there, we immediately have 
		$$y_{m+1} = 0 = y_{m} - x_{m}\Bigg(y_{m-1} - x_{m-1}\bigg(y_{m-2}-x_{m-2}\Big(...\;...-x_{2}(y_{1}-d\cdot x_{1})...\Big)\bigg)\Bigg)$$
		which implies $\mathcal{C}_{m+1}(S, d,x_1,\dots,x_m)$. From there, we immediately get
		$$ g'_m = g'_m(S') \le g_{m+1}(S) \le g_{m+1}.$$
	\end{itemize}
\end{proof}
The above lemma shows in particular how to transform a strategy for $\mathcal{P}'_m$ into a strategy for $\mathcal{P}_{m+1}$ with at least as good cheating probability. This means that we can study $\mathcal{P}'_{m}$ instead of $\mathcal{P}_{m+1}$. The first inequality shows that we do not lose much doing so. 

We also make another change. In order to simplify calculations, we ask Alice to answer at each round $i$ $\tilde{y}_{i} := (-1)^{i+1}y_{i}$ instead of $y_i$. For the protocol, it is totally equivalent to use $\tilde{y}_{i}$ or $y_{i}$ but it allows to avoid all $(-1)$ factors. With this notation, Alice's victory condition $\mathcal{C}'_m(S', d,x_1,\dots,x_m)$ for the protocol becomes:
$$\sum\limits_{i=1}^{m} \left( \tilde{y_{i}}\prod\limits_{j=i+1}^{m}x_{j} \right) = d\prod\limits_{j=1}^{m}x_{j}$$

In the next section, we present a cheating strategy for protocol $\mathcal{P}'_m$.

\subsection{Description of the attack}\label{Section:DescriptionOfTheAttack}
In the previous section, we transformed protocol $\mathcal{P}$ into a slightly modified protocol $\mathcal{P}'$, which has extra symmetries and for which it will be simpler to construct a recursive cheating strategy. In this section, we describe this strategy for $\mathcal{P}'$.

More precisely, we define recursively a strategy with a step of three rounds. To initialize, we consider the following strategy for $\mathcal{P}'_3$: 		
\begin{itemize}
	\item $\mathcal{A}_1$ always outputs $\tilde{y}_{1} = 0$.
	\item $\mathcal{A}_1$ and $\mathcal{A}_2$ perform the optimal strategy for the $\CHSH_Q$ game with inputs $x_1$ and $x_2$. Let $a$ and $b$ be their respective outputs.
	\item $\mathcal{A}_2$ outputs $\tilde{y}_{2} = d \cdot a$ for round $2$ and $\mathcal{A}_1$ outputs $\tilde{y}_{3} = x_{3} \cdot d \cdot b$ for round $3$.
\end{itemize} 
With this strategy, $\mathcal{C'}_{3}$ becomes $x_{3} \cdot d \cdot (a+b - x_{1}\cdot x_{2}) = 0$. Alice wins if $x_{3}=0$, if $d=0$, or if $a+b = x_{1}\cdot x_{2}$. These events are independent, which gives $g'_{3} \geq 1 - \frac{1}{2}(1-\frac{1}{Q})(1-\Cq)$. 

We now describe a strategy for $k+3$ rounds using a strategy for $k$ rounds. We fix a cheating strategy $S'_k$ for Alice for $\mathcal{P}'_k$ and we present a cheating strategy $S'_{k+3}$ for $\mathcal{P}'_{k+3}$. \\ \\
\encad{\begin{center}
		{\textbf{Recursive Description of a cheating strategy $S'_{k+3}$ given $S'_k$}} \end{center}
	\begin{itemize}
		\item Rounds $1$ to $k$: Alice performs the strategy $S'_k$ to get outputs $\tilde{y}_1,\dots,\tilde{y}_k$.
		\item Round $k+1$: Alice always outputs $\tilde{y}_{k+1} = 0$.
		\item Rounds $k+2$ and $k+3$: From round $k+2$, $\mathcal{A}_1$ and $\mathcal{A}_2$ both know $d,x_1,\dots,x_k$. Let 
		$$\eta := d \prod\limits_{j=1}^{k}x_{j} - \sum\limits_{i=1}^{k} ( \tilde{y_{i}}\prod\limits_{j=i+1}^{k}x_{j} ).$$ 
	$\mathcal{A}_{(k+2) \mod 2}$ and $\mathcal{A}_{(k+3) \mod 2}$ perform the optimal strategy for $\CHSH_Q$ on respective inputs $x_{k+2}$ and $x_{k+1}$ to get respective outputs $a$ and $b$. Their outputs of the protocol are respectively $\tilde{y}_{k+2} = \eta \cdot a$ and $\tilde{y}_{k+3} = \eta \cdot b \cdot x_{k+3}$. Notice that if $\eta = 0$, which will correspond to the strategy $S'_k$ succeeding to achieve $\mathcal{C}'_k$, Alice outputs $\tilde{y}_{k+2},\tilde{y}_{k+3} = 0$ independently of $a$ and $b$.
	\end{itemize}
} \\ \\

In the next section, we will prove the cheating probability achieved by this strategy, which will imply our main theorem.

\subsection{Analysis} \label{Section:MainProof}
\begin{lemma} \label{Lemma:Step}
	$\forall k \geq 2$, $g'_{k}$ satisfies :$$ 1 - g'_{k+3}(S'_{k+3}) \le ( 1-\dfrac{1}{Q})\left( 1-\omega(\CHSHQ) \right)(1 - g'_k(S'_k)).$$
	%$$g'_{k+3} \geq \left( 1-\dfrac{1}{Q} \right)\left( 1-\Cq \right)g'_{k} + 1 - \left( 1-\dfrac{1}{Q} \right)\left( 1- \Cq \right)$$
\end{lemma}

\begin{proof}
	We consider $\mathcal{P'}_{k+3}$. Alice's winning condition $\mathcal{C}'_{k+3}$ is:
	$$\sum\limits_{i=1}^{k+3} \left( \tilde{y_{i}}\prod\limits_{j=i+1}^{k+3}x_{j} \right) = d\prod\limits_{j=1}^{k+3}x_{j}$$
	or, by taking apart the last $3$ terms :
	$$\begin{array}{ll}
	&\tilde{y}_{k+3}\\
	+&x_{k+3}\cdot\tilde{y}_{k+2}\\
	+&x_{k+3}\cdot x_{k+2}\cdot\tilde{y}_{k+1}\\
	+&x_{k+3}\cdot x_{k+2}\cdot x_{k+1}\cdot\sum\limits_{i=1}^{k} \left( \tilde{y_{i}}\prod\limits_{j=i+1}^{k}x_{j} \right)\\
	=&x_{k+3}\cdot x_{k+2}\cdot x_{k+1}\cdot d \prod\limits_{j=1}^{k}x_{j}
	\end{array}$$
	
	Recall that $\eta := d \prod\limits_{j=1}^{k}x_{j} - \sum\limits_{i=1}^{k} ( \tilde{y_{i}}\prod\limits_{j=i+1}^{k}x_{j} ) \in \mathbb{F}_{q}$. Using $\eta$, we get :
	$$\mathcal{C'}_{k+3} \Leftrightarrow \tilde{y}_{k+3}
	+ x_{k+3}\cdot\tilde{y}_{k+2} + x_{k+3}\cdot x_{k+2}\cdot\tilde{y}_{k+1} = x_{k+3}\cdot x_{k+2}\cdot x_{k+1}\cdot \eta$$
Recall from our protocol description that $\tilde{y}_{k+2} = \eta \cdot a$ and $\tilde{y}_{k+3} = \eta \cdot b \cdot x_{k+3}$, where $a$ and $b$ are the Alice's outputs of the $\CHSH_Q$ game.
From there, we have
	$$\begin{array}{lll}
	\mathcal{C'}_{k+3} & \Leftrightarrow & \tilde{y}_{k+3}
	+ x_{k+3}\cdot\tilde{y}_{k+2} + x_{k+3}\cdot x_{k+2}\cdot\tilde{y}_{k+1} = x_{k+3}\cdot x_{k+2}\cdot x_{k+1}\cdot \eta\\
	& \Leftrightarrow & x_{k+3} \cdot b \cdot \eta
	+ x_{k+3}\cdot a\cdot \eta + 0 = x_{k+3}\cdot x_{k+2}\cdot x_{k+1}\cdot \eta\\
	& \Leftrightarrow & \eta\cdot x_{k+3}\cdot (a+b-x_{k+1}\cdot x_{k+2}) = 0\\
	& \Leftrightarrow & (x_{k+3} = 0) \vee (\eta = 0) \vee (a+b = x_{k+1} \cdot x_{k+2})
	\end{array} $$
	These $3$ events are independent as :
	\begin{itemize}
		\item $(x_{k+3} = 0)$ only depends on $x_{k+3}$, and happens with probability $\frac{1}{Q}$.
		\item $(\eta = 0)$ only depends on $d, x_{1}, ..., x_{k}$, and happens with probability $g'_k(S'_k)$.
		\item $(a+b = x_{k+1} \cdot x_{k+2})$ only depends on $x_{k+1}$ and $x_{k+2}$ ($\mathcal{A}_{1}$ and $\mathcal{A}_{2}$ optimally play the $\CHSH_Q$ game on inputs $x_{k+1},x_{k+2}$, ignoring any unnecessary information). This happens therefore with probability $\Cq$.
	\end{itemize}
	
	Thus, this particular strategy gives
	$$g'_{k+3}(S'_{k+3}) = \Pr[\mathcal{C'}_{k+3}] = 1-(1-g'_{k}(S'_k))(1-\frac{1}{Q})(1-\omega(\CHSHQ))$$
or equivalently $$ 1 - g'_{k+3}(S'_{k+3}) \le ( 1-\dfrac{1}{Q})\left( 1-\omega(\CHSHQ) \right)(1 - g'_k(S'_k)).$$
\end{proof}
We can now prove our main theorem 
\setcounter{theorem}{1}
\begin{theorem} 
			$\forall m\geq 3$, we have :
			$$g_{m} \geq 1 - \dfrac{1}{2}\left(\left( 1-\dfrac{1}{Q} \right)\left( 1-\omega(\CHSHQ) \right)\right)^{\lfloor \frac{m-1}{3} \rfloor}$$
\end{theorem}
	\begin{proof}
	By iterating the above lemma, we obtain $$1 - g'_{3k}(S'_{3k}) \le  \left(\left( 1-\dfrac{1}{Q} \right)\left( 1-\omega(\CHSHQ) \right)\right)^{k-1}(1 - g'_3(S'_3))$$
	Combining this with the initialization step $g'_3(S_3) \geq 1 - \frac{1}{2}(1-\frac{1}{Q})(1-\Cq)$ gives
	$$g'_{3k} \ge g'_{3k}(S_{3k}) \ge 1 - \dfrac{1}{2}\left(\left( 1-\dfrac{1}{Q} \right)\left( 1-\omega(\CHSHQ) \right)\right)^{k}.$$
	%The initialization step gives $g'_3(S_3) \geq 1 - \frac{1}{2}(1-\frac{1}{Q})(1-\Cq)$. Combining this with the above lemma, we obtain
	%$$ g'_{3k} \ge g'_{3k}(S_{3k}) \ge 1 - \dfrac{1}{2}\left(\left( 1-\dfrac{1}{Q} \right)\left( 1-\omega(\CHSHQ) \right)\right)^{k}.$$
Using the symmetrization lemma (Lemma \ref{Lemma:Symmetrization}), we immediately get 
$$g_{3k+1} \ge g'_{3k} \geq 1 - \dfrac{1}{2}\left(\left( 1-\dfrac{1}{Q} \right)\left( 1-\omega(\CHSHQ) \right)\right)^{k}.$$
If $m$ can be written $m = 3k + 1$ for some $k$, we have
$$g_{m} \geq 1 - \dfrac{1}{2}\left(\left( 1-\dfrac{1}{Q} \right)\left( 1-\omega(\CHSHQ) \right)\right)^{\frac{m-1}{3}}$$
Since $g_m$ is an increasing function, we have for all $m \ge 3$:
$$g_{m} \geq 1 - \dfrac{1}{2}\left(\left( 1-\dfrac{1}{Q} \right)\left( 1-\omega(\CHSHQ) \right)\right)^{\lfloor\frac{m-1}{3}\rfloor}$$
	\end{proof}
	
\section{Generalization}\label{Section:AttackExtension} 
	
In the previous part, we assumed that $\mathcal{A}_{1}$ and $\mathcal{A}_{2}$ can communicate efficiently, very efficiently, meaning that the propagation time $\rho$ is $2$ rounds. With such a propagation time, relativistic constraints ensure that at a given round $k$, Alice cannot use any information concerning the round $k-1$. However, we supposed that she knows everything about the rounds $k-2$ and before. Note that she obviously has access to the information of round $k-2$, because it occurs at the same place than round $k$. 

What happens if $\mathcal{A}_{1}$ and $\mathcal{A}_{2}$ cannot reliably share their knowledge so fast? In this case, the propagation time $\rho$ will be larger, and at any round $k$, Alice knows everything about rounds $1, 2, ..., k-\rho$ with . We use an even propagation time without loss of generality since computations rotate between two places, and Alice always knows what happened at rounds $k-2$, $k-4$, etc. In this situation, we will show that $\mathcal{A}_{1}$ and $\mathcal{A}_{2}$ cannot just play the $\CHSH_Q$ game. They will have to play the $\CHSH_Q^\gamma$ game, for some $\gamma$ that will be specified later.

Another restriction that we do on the cheating players is that $\mathcal{A}_{1}$ and $\mathcal{A}_{2}$ may need some time to determine the bit $d$ they want to decommit to. We call $k_{0}$ the round starting from which both $\mathcal{A}_{1}$ and $\mathcal{A}_{2}$ know if they try to reveal $d=0$ or $d=1$.

In this more practical setting, we propose the following recursive variant of our attack, for $k > k_0$, for any propagation time $\rho \ge 2$. \\ \\
\encad{\begin{center}
		{\textbf{Recursive Description of a cheating strategy $S_{k+\rho+1}$ given $S_k$}} \end{center}
	\begin{itemize}
		\item Rounds $1$ to $k$: Alice performs the strategy $S_k$ to get outputs $\tilde{y}_1,\dots,\tilde{y}_k$.
		\item Rounds $k+1$ to $k+\rho - 1$: Alice  outputs $\tilde{y}_{k+1},\dots,\tilde{y}_{k+\rho - 1} = 0$.
		\item Rounds $k+\rho$ and $k+\rho+1 $: From round $k+\rho$, $\mathcal{A}_1$ and $\mathcal{A}_2$ both know $d,x_1,\dots,x_k$. Let 
		$$\eta := d \prod\limits_{j=1}^{k}x_{j} - \sum\limits_{i=1}^{k} ( \tilde{y_{i}}\prod\limits_{j=i+1}^{k}x_{j} ).$$ $\mathcal{A}_{1}$ also knows $X = \prod_{j \textrm{ odd } : \ k +1 \leq j \leq k + \rho} \ x_j$ and $\mathcal{A}_2$ knows $Y = \prod_{j \textrm{ even } : \ k +1 \leq j \leq k + \rho} \ x_j$. 
$\mathcal{A}_{(k+\rho) \mod 2}$ and $\mathcal{A}_{(k+\rho+1) \mod 2}$ perform the optimal strategy for $\CHSH_Q^\gamma$ with $\gamma := 1 - (1-\frac{1}{Q})^{\frac{\rho}{2}}$ on respective inputs $X$ and $Y$ to get respective outputs $a$ and $b$. Their outputs of the protocol are respectively $\tilde{y}_{k+\rho} = \eta \cdot a$ and $\tilde{y}_{k+\rho+1} = \eta \cdot b \cdot x_{k+\rho+1}$.
	\end{itemize}
} \\ \\

\begin{lemma} \label{Lemma:StepExtension}
	$\forall k \geq k_{0}$, we have 
	$$g'_{k+\rho+1} \geq \left( 1-\dfrac{1}{Q} \right)\left( 1- \omega(\CHSHQ) \right)g'_{k} + 1 - \left( 1-\dfrac{1}{Q} \right)\left( 1-\omega(\CHSHQ) \right)$$
\end{lemma}

\begin{proof}
	This demonstration will be similar to Lemma \ref{Lemma:Step}. We consider the cheating strategy described above. Alice's winning condition $\mathcal{C}'_{k+\rho+1}$ is:
	$$\sum\limits_{i=1}^{k+\rho+1} \left( \tilde{y_{i}}\prod\limits_{j=i+1}^{k+\rho+1}x_{j} \right) = d\prod\limits_{j=1}^{k+\rho+1}x_{j}$$
	or, by separating the last $\rho + 1$ terms :
	$$\begin{array}{ll}
	&\tilde{y}_{k+\rho+1}\\
	+&x_{k+\rho+1}\cdot\tilde{y}_{k+\rho}\\
	+&0\\
	\vdots\\
	+&0\\
	+&\left(\prod\limits_{j=k+1}^{k+\rho+1}x_{j}\right) \sum\limits_{i=1}^{k} \left( \tilde{y_{i}}\prod\limits_{j=i+1}^{k}x_{j} \right)\\
	=&\left(\prod\limits_{j=k+1}^{k+\rho+1}x_{j}\right) d \prod\limits_{j=1}^{k}x_{j}
	\end{array}$$
	Recall that $\eta := d \prod\limits_{j=1}^{k}x_{j} - \sum\limits_{i=1}^{k} ( \tilde{y_{i}}\prod\limits_{j=i+1}^{k}x_{j} )$, which allows to simplify $\mathcal{C}'_{k+\rho+1}$ as follows
	
	 \begin{align*}
	 \mathcal{C}'_{k+\rho+1} & \Leftrightarrow \tilde{y}_{k+\rho+1} + x_{k+\rho+1}\cdot\tilde{y}_{k+\rho} = \left(\prod\limits_{j=k+1}^{k+\rho+1}x_{j}\right) \cdot \eta \\
	 & \Leftrightarrow \tilde{y}_{k+\rho+1} + x_{k+\rho+1}\cdot\tilde{y}_{k+\rho} = X \cdot Y \cdot \eta.
	 \end{align*}

	In her cheating strategy, Alice answers $\tilde{y}_{k+\rho} = a \cdot \eta$ and $\tilde{y}_{k+\rho+1} = x_{k+\rho+1} \cdot b \cdot \eta$, where $a$ and $b$ are the Alices' answers when playing the $\CHSH_Q^\gamma$ game, on inputs $X$ and $Y$. Thus 
	
	$$\mathcal{C}'_{k+\rho+1} \Leftrightarrow  (x_{k+\rho+1} = 0) \vee (\eta = 0) \vee (a+b = XY)$$
	
	These three events are independent. The first one occurs with probability $\frac{1}{Q}$, the second one with probability $g'_{k}$. For the third one, notice that $X$ is a product of $\frac{\rho}{2}$ uniformly random number in $\F_Q$. Therefore, we have $\Pr[X = 0] = 1 - (1-\frac{1}{Q})^{\frac{\rho}{2}} = \gamma$ and for any $z \in \F_Q^*$, $\Pr[X = z] = \frac{1 - \gamma}{Q-1}$. $Y$ satisfies the same probability distribution. Therefore, $\Pr[a+b = XY]$ is exactly the probability of winning the $\CHSH_Q^\gamma$ game using its optimal strategy.
	
	This gives :
	$$g'_{k+\rho+1} \geq 1-\left( 1-\dfrac{1}{Q} \right)(1-g'_{k})\left( 1- \omega(\CHSH_Q^\gamma) \right)$$
	Then using Lemma \ref{Lemma:CHSH_Q^delta} :
	$$g'_{k+\rho+1} \geq 1-\left( 1-\dfrac{1}{Q} \right)(1-g'_{k})\left( 1- \omega(\CHSH_Q) \right)$$
	i.e.
	$$g'_{k+\rho+1} \ge \left( 1-\dfrac{1}{Q} \right)\left( 1- \omega(\CHSH_Q) \right)g'_{k} + 1 - \left( 1-\dfrac{1}{Q} \right)\left( 1-\omega(\CHSH_Q) \right)$$
\end{proof}
\begin{theorem}
	For any $k_0 \ge 1$ and $\rho \ge 2$, for any $m \ge k_0 + \rho + 1$, we have 
	$$ g_m \ge 1 - \dfrac{1}{2}\left(( 1-\dfrac{1}{Q})\left( 1-\omega(\CHSHQ) \right)\right)^{\frac{m - k_0 - 1}{\rho + 1}} $$
	\end{theorem}
	
	\begin{proof}
		We use the recursive inequality from Lemma \ref{Lemma:StepExtension}, and the trivial initialization $g'_{k_{0}} \geq \frac{1}{2}$. This gives us $\forall k \geq k_0$, we have
		$$g'_{k_{0}+k(\rho + 1)} \geq 1 - \dfrac{1}{2}\left(( 1-\dfrac{1}{Q})\left( 1-\omega(\CHSHQ) \right)\right)^{k},$$
		and by using Lemma \ref{Lemma:Symmetrization}
		$$
		g_{k_{0}+k(\rho + 1) + 1} \geq 1 - \dfrac{1}{2}\left(( 1-\dfrac{1}{Q})\left( 1-\omega(\CHSHQ) \right)\right)^{k}.$$
		If $m$ can be written $m = k_{0}+k(\rho + 1) + 1$, we have $k = \frac{m - k_0 - 1}{\rho + 1}$ and 
		$$ g_m \ge 1 - \dfrac{1}{2}\left(( 1-\dfrac{1}{Q})\left( 1-\omega(\CHSHQ) \right)\right)^\frac{m - k_0 - 1}{\rho + 1}.$$
		In order, to conclude, notice that $g_m$ is an increasing function in $m$. We can therefore conclude that 
		$$ g_m \ge 1 - \dfrac{1}{2}\left(( 1-\dfrac{1}{Q})\left( 1-\omega(\CHSHQ) \right)\right)^{\lfloor\frac{m - k_0 - 1}{\rho + 1}\rfloor}. $$
	\end{proof}
\bibliography{carbc9}
\bibliographystyle{alpha}

\end{document}